\documentclass[twocolumn,%superscriptaddress,
showpacs,preprintnumbers,amsmath,amssymb]{revtex4}
\usepackage{graphicx}% Include figure files
\usepackage{amsthm}

\newtheorem{Thm}{Theorem}
\newtheorem{Lem}[Thm]{Lemma}
\newtheorem{Prop}[Thm]{Proposition}
\newtheorem{Cor}[Thm]{Corollary}
\theoremstyle{definition}
\newtheorem{Def}[Thm]{Definition}
\newtheorem{Rem}[Thm]{Remark}

\newcommand{\ket}[1]{\left|{#1}\right\rangle}
\newcommand{\bra}[1]{\left\langle{#1}\right|}

\newcommand{\Tr}{\mathop{\mathrm{Tr}}\nolimits}
\def\Complex{\mathbb{C}}
\def\Real{\mathbb{R}}

\newcommand{\C}{\mbox{$\mathcal{C}$}}
\newcommand{\M}{\mbox{$\mathcal{M}$}}

\newcommand*{\cI}{\mathcal{I}}
\newcommand*{\cL}{\mathcal{L}}
\newcommand*{\cN}{\mathcal{N}}

\begin{document}

\title{
Quantum non-signalling assisted zero-error classical capacity of qubit channels
}

\author{Jeonghoon Park}\email{zucht@khu.ac.kr}
\affiliation{
Department of Mathematics and Research Institute for Basic Sciences,
Kyung Hee University, Seoul 130-701, Korea
}
\author{Soojoon Lee}\email{level@khu.ac.kr}
\affiliation{
Department of Mathematics and Research Institute for Basic Sciences,
Kyung Hee University, Seoul 130-701, Korea
}

\date{\today}

%%%%%%%%%%%%%%%%%%%%%%%%%%%%%%%%%%%%%%%%%%%%
%%                                                            Abstract                                                    %%
%%%%%%%%%%%%%%%%%%%%%%%%%%%%%%%%%%%%%%%%%%%%
\begin{abstract}
In this paper, 
we explicitly evaluate the one-shot quantum non-signalling assisted zero-error classical capacities 
$\M_0^{\mathrm{QNS}}$ for qubit channels.
In particular, 
we show that for nonunital qubit channels, $\M_0^{\mathrm{QNS}}=1$, 
which implies that in the one-shot setting, 
nonunital qubit channels cannot transmit any information with zero probability of error 
even when assisted by quantum non-signalling correlations.
Furthermore, 
we show that for qubit channels, $\M_0^{\mathrm{QNS}}$ equals to 
the one-shot entanglement-assisted zero-error classical capacities.
This means that for a single use of a qubit channel, 
quantum non-signalling correlations are not more powerful than shared entanglement.
\end{abstract}


\maketitle

%%%%%%%%%%%%%%%%%%%%%%%%%%%%%%%%%%%%%%%%%%%%
%%                                                      Introduction                                                    %%
%%%%%%%%%%%%%%%%%%%%%%%%%%%%%%%%%%%%%%%%%%%%
\section{Introduction}
%% zero-error channel capacity %%
While the ordinary channel capacity allows errors
which can be made arbitrarily small in sufficiently many channel uses,
the zero-error channel capacity does not allow any error.
Thus, in many situations that no error is tolerated or 
only a small number of channel uses are available,
it is important to consider the zero-error channel capacity~\cite{KO98}.

%% entanglement & NS correlations as resource %%
Additional resources as correlations between sender and receiver in a channel
can be used for communication,
and they may increase its channel capacity.
% entanglement
Indeed, it has been known that shared entanglement can increase zero-error capacities 
of both classical channels~\cite{CLMW10,LMMOR12} and quantum channels~\cite{Duan09}.
However, 
we cannot send information by only using the shared entanglement 
without any additional classical/quantum channel.
For example,
the dense coding~\cite{BW92} and the quantum teleportation~\cite{BBCJPW93} use 
shared entanglement as a resource,
but both of them need a classical/quantum channel %-s? 
through which information can be sent.
% QNS
As more general resources, 
classical and quantum non-signalling (QNS) correlations have been introduced 
in Ref.~\cite{CLMW11} and Ref.~\cite{DW14}, respectively.
In particular, 
QNS correlations can be viewed as bipartite quantum channels
from $A_1 \otimes B_1$ to $A_2 \otimes B_2$
%$\Omega : \cL(A_1)\otimes\cL(B_1) \rightarrow \cL(A_2)\otimes\cL(B_2)$
through which the parties $A$ and $B$ cannot send any information to each other.
%Thus, QNS correlations can be thought as useless channels in the sense, %추가된것
Thus, it is said to be quantum non-signalling correlations,
and shared entanglement and classical non-signalling correlations can be regarded
as special cases of QNS correlations. 

%% qubit chs. %%
We here take into account
the QNS assisted zero-error classical capacity of qubit channels.
For qubit channels, 
the one-shot zero-error classical capacity cannot be superactivated~\cite{PL12}
(in fact, it also holds for qutrit channels~\cite{SS15}),
although the superactivation of quantum channels is 
a peculiar quantum effect with no classical analogue~\cite{Duan09,CCH11}.
In addition,
regularization is not necessary 
for the (asymptotic) zero-error classical capacity of qubit channels
and even for the entanglement-assisted zero-error classical capacity~\cite{DSW13}.
From these properties,
we may think that
qubit channels have 
somewhat different properties
from higher dimensional cases.

%% our results %%
In this paper, 
we evaluate explicitly the one-shot QNS assisted zero-error classical capacities 
$\M_0^{\mathrm{QNS}}$ for qubit channels.
For unital qubit channels, %추가된것
it is not so hard to calculate $\M_0^{\mathrm{QNS}}$
unlike nonunital case. 
%In particular, 
Here, we show that $\M_0^{\mathrm{QNS}}=1$ for any nonunital qubit channel. 
In other words, 
nonunital qubit channels cannot transmit any information 
with zero probability of error 
when assisted by entanglement,
or even when assisted by QNS correlations
in the one-shot setting.
Moreover, 
we show that for qubit channels, $\M_0^{\mathrm{QNS}}$ equals to 
the one-shot entanglement-assisted zero-error classical capacity $\M_0^{\mathrm{SE}}$.
This means that 
QNS correlations are not more powerful than shared entanglement 
for a single use of a qubit channel.

%% organization %%
This paper is organized as follows.
In Section~\ref{sec:one-shot},
we evaluate the one-shot QNS assisted zero-error classical capacities for qubit channels.
In Section~\ref{sec:main}, 
we show that for qubit channels, 
$\M_0^{\mathrm{QNS}}$ equals to $\M_0^{\mathrm{SE}}$.
Finally, we summarize our results and conclude with discussion in Section~\ref{sec:conclusion}.

%%%%%%%%%%%%%%%%%%%%%%%%%%%%%%%%%%%%%%%%%%%%
%%                                    M_0^QNS for qubit chs.                                                     %%
%%%%%%%%%%%%%%%%%%%%%%%%%%%%%%%%%%%%%%%%%%%%
\section{The one-shot QNS assisted zero-error classical capacities of qubit channels}
\label{sec:one-shot}
In this section, 
we explicitly calculate the one-shot QNS assisted zero-error classical capacities of 
Pauli channels and nonunital qubit channels using the formula in Ref.~\cite{DW14}.
Since any unital qubit channel is unitarily equivalent to a Pauli channel~\cite{King02},
it is sufficient to consider Pauli channels rather than all unital qubit channels.

%% formular of M_0^QNS & C_0^QNS %%
It is known~\cite{DW14} that
one can calculate 
the one-shot QNS assisted zero-error classical capacities $\M_0^{\mathrm{QNS}}$ 
by some semi-definite programming (SDP) as follows.

\begin{Lem}\label{Lemma:M_0^QNS}
For a given quantum channel $\cN$ from a quantum system $A$ to a quantum system $B$, 
let $\Upsilon(\cN)$ be the quantity 
obtained from the following SDP:
\begin{equation}\label{eq:Upsilon}
\Upsilon(\cN) = \max \Tr S_A,  
\end{equation}
where the maximum is taken over all $S_A$'s satisfying
\begin{align}
&0 \le U_{AB} \le S_A \otimes I_B, \label{eq:cons1} \\
&\Tr_A U_{AB} = I_B, \label{eq:cons2} \\
&\Tr P_{AB}(S_A \otimes I_B - U_{AB}) = 0. \label{eq:cons3}
\end{align}
Here, 
$P_{AB}$ is the projection onto the support of the Choi-Jamio{\l}kowski (CJ) matrix 
$(\cI \otimes \cN)\ket\Phi_{AB}\bra\Phi$ of $\cN$,
where $\ket\Phi_{AB} = \sum_{k}\ket{k}_A\ket{k}_B$ be the unnormalized maximally entangled state.
Then $\M_0^{\mathrm{QNS}}(\cN)$ is 
the integer part $\lfloor \Upsilon(\cN) \rfloor$ of the quantity $\Upsilon(\cN)$.
\end{Lem}

We note that since $\Upsilon(\cN)$ essentially depends on the projection $P_{AB}$,
$\Upsilon(\cN)$ can be also denoted by $\Upsilon(P_{AB})$.

%%%                                         Pauli channels                                              %%%
%%%%%%%%%%%%%%%%%%%%%%%%%%%%%%%%%%%%%%%%%
\subsection{Pauli channels}
\label{subsec:Pauli}
We first consider Pauli channels defined as follows.

%% def. of Pauli chs. %%
\begin{Def}\label{Def:Pauli}
For a probability distribution $\{p_{ij}\}$,
the {\em Pauli channel} $\cN^{P}$ is defined by
\begin{equation}\label{eq:Pauli}
\cN^{P}(\rho) \equiv \sum_{i,j=0}^{1} p_{ij}X^iZ^j \rho (X^iZ^j)^\dagger,
\end{equation}
where
$X \equiv \sum_{j=0}^{1} \ket{j \oplus 1}\bra{j}$ 
and
$Z \equiv \sum_{j=0}^{1} (-1)^{j}\ket{j}\bra{j}$.
\end{Def}

%% Upsilon and C_0^{QNS }for Pauli chs. %%
\begin{Lem}\label{Lem:max.e.}
Assume that both $A$ and $B$ are two-dimensional quantum systems, 
and let $P_{AB}$ be the projection defined as
\begin{equation}\label{eq:max.e.proj}
P_{AB} = \sum_{i=1}^{k} \ket{\phi_i}_{AB}\bra{\phi_i},
\end{equation}
where $\ket{\phi_i}$'s are mutually orthogonal and maximally entangled states.
Then $\Upsilon(P_{AB}) = 4/k$.
\end{Lem}
\begin{proof}
For any $S_A$ and $U_{AB}$ satisfying
the constraints (\ref{eq:cons1}), (\ref{eq:cons2}), and (\ref{eq:cons3}),
\[
k\Tr S_A = 2\Tr P_{AB}U_{AB} \leq 2\Tr U_{AB} = 2\Tr I_B = 4.
\]
Thus we can obtain the inequality $\Tr S_A \leq 4/k$. 
We now take $(2/k)I_A$ as $S_A$ and $(2/k)P_{AB}$ as $U_{AB}$.
%$T_A \equiv (2/k)I_A$ and $V_{AB} \equiv (2/k)P_{AB}$.
Then it can be readily seen that
the constraints (\ref{eq:cons1}), (\ref{eq:cons2}), and (\ref{eq:cons3}) hold, 
and $\Tr ((2/k)I_A) = 4/k$.
Hence, $\Upsilon(P_{AB}) = 4/k$.
\end{proof}

Using the above lemma,
we can obtain explicit values of $\Upsilon$ for Pauli channels.

\begin{Thm}\label{Prop:Pauli ch.}
Let $\cN^{P}$ be the Pauli channel
with a probability distribution $\{p_{ij}\}$.
Then $\Upsilon(\cN^{P}) = 4/k$,
% and $\M_0^{\mathrm{QNS}}(\cN^{P}) = \lfloor4/k\rfloor$,
where k is the number of nonzero probabilities $p_{ij}$.
\end{Thm}
\begin{proof}
Since the CJ matrix of $\cN^{P}$ is
\[
\sum_{i,j=0}^{1}p_{ij}
\left(I \otimes X^iZ^j\right) \ket{\Phi}\bra{\Phi} \left(I \otimes X^iZ^j\right)^\dagger,
\]
and $\left(I \otimes X^iZ^j\right)\ket{\Phi}$'s are
mutually orthogonal and (unnormalized) maximally entangled states,
the projection $P_{AB}$ associated with the channel $\cN^{P}$ is
the sum of $k$ mutually orthogonal and maximally entangled states as one-dimensional projections
which is of the form in Eq.~(\ref{eq:max.e.proj}).
Thus by Lemma~\ref{Lem:max.e.}, 
it is clear that 
\[
\Upsilon(\cN^{P})=\Upsilon(P_{AB})=4/k.
\]
\end{proof}

The (asymptotic) QNS assisted zero-error classical capacity $\C_0^{\mathrm{QNS}}(\cN)$ 
of a quantum channel $\cN$ is defined by
\begin{equation}\label{eq:C_0^QNS}
\C_0^{\mathrm{QNS}}(\cN) = \lim_{n \rightarrow \infty}\frac{1}{n}\log{\Upsilon(\cN^{\otimes n})}.
\end{equation}
Since it is not hard to show that 
$\log\Upsilon$ is additive for Pauli channels,
we directly have the following corollary.

\begin{Cor}\label{Cor:QNS_Pauli}
For Pauli channels $\cN^{P}$, 
$\C_0^{\mathrm{QNS}}$ is additive, and
$\C_0^{\mathrm{QNS}}(\cN^{P}) = \log(4/k)$, 
where k is the number of nonzero probabilities $p_{ij}$.
\end{Cor}

\begin{Rem}\label{Rem:g.Pauli}
The arguments in this subsection can be applied to 
higher dimensional Pauli channels as well.
Then we can also calculate the QNS assisted zero-error classical capacities 
for generalized Pauli channels.
\end{Rem}

%%%                                Nonunital qubit channels                                      %%%
%%%%%%%%%%%%%%%%%%%%%%%%%%%%%%%%%%%%%%%%%
\subsection{Nonunital qubit channels}
\label{subsec:nonunital}
We now show that the value of $\Upsilon$ is one for nonunital qubit channels.
This means that in the one-shot setting, 
nonunital qubit channels cannot send any message with zero probability of error 
even when assisted by QNS correlations.

\begin{Thm}\label{Prop:Upsilon_nonunital}
For any nonunital qubit channel $\cN$, $\Upsilon(\cN) = 1$.
\end{Thm}
\begin{proof}
Assume that $\Upsilon(\cN) > 1$,
then we will show that $\cN$ must be unital.
Let $J_{AB}$ be the CJ matrix of a qubit channel $\cN$.
We note that 
$\Tr_B{J_{AB}} = I_A$, 
since $\cN$ is a quantum channel, and 
that $\Tr_A{J_{AB}} = I_B$ if and only if $\cN$ is unital.
Moreover,
$\mathrm{rank}(J_{AB}) < 4$, 
since otherwise $P_{AB} = I_{AB}$, 
and so $\Upsilon(P_{AB}) = 1$.

We first suppose that $\mathrm{rank}(J_{AB}) = 1$.
Then $J_{AB} = \alpha\ket{\phi}\bra{\phi}$ 
for some $\alpha > 0$ and some state $\ket{\phi}$.
Since $\Tr_B{J_{AB}} = I_A$, 
$\ket{\phi}$ is maximally entangled.
Thus, $\Tr_A{J_{AB}} = I_B$, 
and so $\cN$ is unital.

We now suppose that $\mathrm{rank}(J_{AB}) = 3$.
Then $P_{AB} = I_{AB}-\ket{\psi}\bra{\psi}$ 
for some state $\ket{\psi}$.
Since $\Upsilon(P_{AB}) > 1$, 
there exist $S_A$ with $\Tr{S_A}>1$ and $U_{AB}$ satisfying the constraints
(\ref{eq:cons1}), (\ref{eq:cons2}), and (\ref{eq:cons3}).
Then $S_A \otimes I_B - U_{AB} = \alpha\ket{\psi}\bra{\psi}$ for some $\alpha > 0$.
By tracing out the system A,
we can see that $\ket{\psi}$ must be maximally entangled.
Up to local unitary,
we may assume that $P_{AB} = I_{AB}-\ket{\phi^+}\bra{\phi^+}$
and $J_{AB} = \sum_{i=1}^3\alpha_i\ket{\phi_i}\bra{\phi_i}$, 
where $\sum_{i}\alpha_i=2$ with $\alpha_i > 0$, 
%$\ket{\phi^\pm} = (\ket{00}\pm\ket{11})/\sqrt{2}$, and 
%$\ket{\phi_i} \in \mathrm{span}\{\ket{\phi^-}, \ket{01}, \ket{10}\}$.
$\ket{\phi^+} = (\ket{00}+\ket{11})/\sqrt{2}$, and 
$\ket{\phi_i}$ are orthogonal to $\ket{\phi^+}$.
Since $\Tr_A{\ket{\psi}\bra{\psi}} + (\Tr_B{\ket{\psi}\bra{\psi}})^T = I$ %레퍼런스에?
for any state $\ket{\psi}$ orthogonal to $\ket{\phi^+}$, 
$\Tr_A{J_{AB}} = 2I - (\Tr_B{J_{AB}})^T = I$. 
Thus, $\cN$ is unital.

Let us assume that $\mathrm{rank}(J_{AB}) = 2$ 
in the rest of this proof.
%For the case of $\mathrm{rank}(P_{AB}) = 2$,
%the proof is more complicated.
Without loss of generality, let
\[
J_{AB} = r\ket{\psi_1}\bra{\psi_1} + (2-r)\ket{\psi_2}\bra{\psi_2},
\]
where
$0 < r \le 1$,
and
$\ket{\psi_1} = \sqrt{\lambda}\ket{00}+\sqrt{1-\lambda}\ket{11}$
and $\ket{\psi_2} = \sum_{i,j=0}^{1}a_{ij}\ket{ij}$
are orthogonal states
%$\ket{\psi_1}\perp\ket{\psi_2}$,
for $0 \le \lambda \le 1/2$.
Since $\ket{\psi_1}$ and $\ket{\psi_2}$ are orthogonal,
\begin{equation}\label{eq:a_{11}}
a_{11}= -\sqrt{\frac{\lambda}{1-\lambda}}a_{00}.
\end{equation}
Since $\Tr_B(J_{AB}) = I_A$,
\begin{equation}\label{eq:zero_term}
a_{00}a_{10}^*+a_{01}a_{11}^* = 0
\end{equation}
and
%$r\lambda+(2-r)(|a_{00}|^2+|a_{01}|^2) = 1 = r(1-\lambda)+(2-r)(|a_{10}|^2+|a_11}|^2)$.
\begin{equation}\label{eq:diag_equality}
r(1-2\lambda) = (2-r)(2|a_{00}|^2+2|a_{01}|^2-1).
\end{equation}
Since $0 < r \le 1$ and $\lambda \le 1/2$,
it follows from Eq.~(\ref{eq:diag_equality}) that
\begin{equation}\label{eq:sum_coeff}
|a_{00}|^2+|a_{01}|^2 \ge 1/2.
\end{equation}

We divide the remaining proof into three cases
according to the values of $a_{01}$ and $a_{00}$;
(Case~1)~$a_{01}=0$,
(Case~2)~$a_{00}=0$, and
(Case~3)~$a_{00} \ne 0$ and $a_{01} \ne 0$.

(Case 1)
From Eqs.~(\ref{eq:zero_term}) and~(\ref{eq:sum_coeff}),
$a_{10}=0$.
Then $\ket{\psi_2} = a_{00}\ket{00}+a_{11}\ket{11}$, and 
$\Tr_A{J_{AB}} = \Tr_B{J_{AB}} = I$.

(Case 2)
By Eq.~(\ref{eq:a_{11}}),
$a_{11}=0$.
%Then $\ket{\psi_1} = \sqrt{\lambda}\ket{00}+\sqrt{1-\lambda}\ket{11}$ and
%$\ket{\psi_2} = a_{01}\ket{01}+a_{10}\ket{10}$.
Then we can choose an orthonormal basis $\{\ket{\psi_3}, \ket{\psi_4}\}$
for the subspace orthogonal to $P_{AB}$,
where
$\ket{\psi_3} = \sqrt{1-\lambda}\ket{00}-\sqrt{\lambda}\ket{11}$ and
$\ket{\psi_4} = a_{10}^*\ket{01}-a_{01}^*\ket{10}$.
Let $S_A$ and $U_{AB}$ satisfy the constraints
(\ref{eq:cons1}), (\ref{eq:cons2}), and (\ref{eq:cons3}).
By the constraints (\ref{eq:cons1}) and (\ref{eq:cons3}),
\begin{eqnarray}\label{eq:feasible1}
S_A \otimes I_B - U_{AB}
&=& a\ket{\psi_3}\bra{\psi_3} + b\ket{\psi_4}\bra{\psi_4} \nonumber \\
&& + c\ket{\psi_3}\bra{\psi_4} + c^*\ket{\psi_4}\bra{\psi_3}
\end{eqnarray}
for some $a,b \ge 0$ and $c \in \Complex$.
Tracing out the system $A$ on both sides of Eq.~(\ref{eq:feasible1}),
by the constraint (\ref{eq:cons2}),
we can obtain
%\[
%a(1-\lambda)+b|a_{01}|^2 = a\lambda+b|a_{10}|^2.
%\]
%Then 
$a(1-2\lambda) = b(1-2|a_{01}|^2)$ and $|a_{01}|^2 \le 1/2$.
By Eq.~(\ref{eq:sum_coeff}),
$|a_{01}|^2 = 1/2$,
and so $\ket{\psi_2}$ is maximally entangled.
Moreover,
from Eq.~(\ref{eq:diag_equality}),
%$\lambda =1/2$,
%and so 
we can see that $\ket{\psi_1}$ is also maximally entangled.
Thus, $\Tr_A{J_{AB}} = I_B$.

(Case 3)
Since $\Upsilon(\cN) > 1$,
there are $S_A$ with $\Tr{S_A}>1$ and $U_{AB}$ satisfying
the constraints (\ref{eq:cons1}), (\ref{eq:cons2}), and (\ref{eq:cons3}).
%s(\ref{eq:cons1}, \ref{eq:cons2}, \ref{eq:cons3})
By the constraints (\ref{eq:cons1}) and (\ref{eq:cons3}),
\begin{equation}\label{eq:feasible2}
S_A \otimes I_B - U_{AB}=
\alpha\ket{\psi_3}\bra{\psi_3} + \beta\ket{\psi_4}\bra{\psi_4},
\end{equation}
where $\{\ket{\psi_3}, \ket{\psi_4}\}$ is an orthonormal basis
for the subspace orthogonal to $P_{AB}$
and $\alpha, \beta \ge 0$.
%$\ket{\psi_3}$ and $\ket{\psi_4}$ are orthogonal to $P_{AB}$,
%$\ket{\psi_3}\perp\ket{\psi_4}$, 
%and $\alpha, \beta \ge 0$.
Let $\ket{\psi_3} = \sum_{i,j=0}^{1}b_{ij}\ket{ij}$
and $\ket{\psi_4} = \sum_{i,j=0}^{1}c_{ij}\ket{ij}$.
Since $\ket{\psi_3}$ is orthogonal to $\ket{\psi_1}$ and $\ket{\psi_2}$,
%\begin{equation}\label{eq:b_{11}}
%b_{11} = -\sqrt{\frac{\lambda}{1-\lambda}}b_{00}
%\end{equation}
%and
%\begin{equation}\label{eq:coeffs}
%a_{01}^*b_{01} = - \frac{a_{00}^*b_{00}}{1-\lambda} - a_{10}^*b_{10}.
%\end{equation}
we can note that
\begin{eqnarray}
b_{11} &=& -\sqrt{\frac{\lambda}{1-\lambda}}b_{00}, 
\label{eq:ortho_1} \\
a_{01}^*b_{01} &=& -\frac{a_{00}^*b_{00}}{{1-\lambda}} - a_{10}^*b_{10}.
\label{eq:ortho_2}
\end{eqnarray}
Similar equations hold for the coefficients of $\ket{\psi_4}$.

Tracing out the system $A$ on both sides of Eq.~(\ref{eq:feasible2}),
by the constraint (\ref{eq:cons2}),% and Eq.~(\ref{eq:b_{11}}),
the following three equalities can be obtained.
%\begin{eqnarray}
%(\Tr{S_A}-1)I_B &=& \alpha
%\begin{pmatrix}
%|b_{00}|^2+|b_{10}|^2 & b_{00}b_{01}^* - \sqrt{\frac{\lambda}{1-\lambda}}b_{00}^*b_{10} \\
%b_{00}^*b_{01} - \sqrt{\frac{\lambda}{1-\lambda}}b_{00}b_{10}^* & \frac{\lambda}{1-\lambda}|b_{00}|^2+|b_{01}|^2
%\end{pmatrix} \nonumber \\
%&&+ \beta\Tr_{A}{\ket{\psi_4}\bra{\psi_4}}.\label{eq:partial_Tr}
%\end{eqnarray}
\begin{eqnarray}
\Tr{S_A}-1 &=& \alpha(|b_{00}|^2+|b_{10}|^2) \nonumber \\
&& + \beta(|c_{00}|^2+|c_{10}|^2),~\label{eq:(1,1)} \\
\Tr{S_A}-1 &=&
\alpha\left(\frac{\lambda |b_{00}|^2}{1-\lambda}+|b_{01}|^2\right) \nonumber \\
&& + \beta\left(\frac{\lambda |c_{00}|^2}{1-\lambda}+|c_{01}|^2\right),~\label{eq:(2,2)} \\
0 &=&
\alpha\left(b_{00}b_{01}^* - \sqrt{\frac{\lambda}{1-\lambda}}b_{00}^*b_{10}\right) \nonumber \\
&& + \beta\left(c_{00}c_{01}^* - \sqrt{\frac{\lambda}{1-\lambda}}c_{00}^*c_{10}\right)~\label{eq:(1,2)}.
\end{eqnarray}
%Multiplying both sides in Eq.~(\ref{eq:(1,2)}) by $a_{00}^*a_{01}$,
%we obtain from Eq.~(\ref{eq:coeffs}) that
%\begin{eqnarray}
%0 &=& \alpha\left(\frac{|a_{00}|^2}{1-\lambda}|b_{00}|^2 +2\Re(a_{00}^*a_{10}b_{00}b_{10}^*)\right) \nonumber \\
%&& +\beta\left(\frac{|a_{00}|^2}{1-\lambda}|c_{00}|^2 +2\Re(a_{00}^*a_{10}c_{00}c_{10}^*)\right)~\label{eq:(1,2)`}.
%\end{eqnarray}
%Subtracting Eq.~(\ref{eq:(1,1)}) multiplied by $\lambda|a_{00}|^2|a_{01}|^2$
%from Eq.~(\ref{eq:(2,2)}) multiplied by $(1-\lambda)|a_{00}|^2|a_{01}|^2$,
%we obtain
%\begin{eqnarray}~\label{eq:subtracting}
%(1-2\lambda)|a_{00}|^2|a_{01}|^2(\Tr{S_A}-1) = \nonumber \\
%\alpha\left(\frac{|a_{00}|^2}{1-\lambda}|b_{00}|^2
%+2\Re(a_{00}^*a_{10}b_{00}b_{10}^*]\right) \nonumber \\
%+\alpha\left(
%(1-\lambda)|a_{10}|^2-\lambda|a_{01}|^2
%\right)|a_{00}|^2|b_{10}|^2 \nonumber \\
%+\beta\left(\frac{|a_{00}|^2}{1-\lambda}|c_{00}|^2
%+2\Re(a_{00}^*a_{10}c_{00}c_{10}^*)\right) \nonumber \\
%+\beta\left(
%(1-\lambda)|a_{10}|^2-\lambda|a_{01}|^2
%\right)|a_{00}|^2|c_{10}|^2.
%\end{eqnarray}
%Then Eq.~(\ref{eq:subtracting}) becomes from~(\ref{eq:(1,2)`}) that
%Adding Eq.~(\ref{eq:(1,2)}) multiplied by $a_{00}^*a_{01}$ to Eq.~(\ref{eq:subtracting})
%Adding Eq.~(\ref{eq:(1,1)}) multiplied by $-\lambda |a_{01}|^2$
%and Eq.~(\ref{eq:(2,2)}) multiplied by $(1-\lambda) |a_{01}|^2$
%to Eq.~(\ref{eq:(1,2)}) multiplied by $a_{00}^*a_{01}$,
Multiply both sides of 
Eq.~(\ref{eq:(1,1)}), Eq.~(\ref{eq:(2,2)}), and Eq.~(\ref{eq:(1,2)})
by $-\lambda |a_{01}|^2$, $(1-\lambda) |a_{01}|^2$, and $a_{00}^*a_{01}$, respectively.
Then we can obtain
from Eqs.~(\ref{eq:a_{11}}), (\ref{eq:zero_term}), (\ref{eq:ortho_1}), and (\ref{eq:ortho_2}) that
their sum becomes 
\[
(1-2\lambda)|a_{01}|^2(\Tr{S_A}-1) =0.
\]
%and we can obtain
%from Eqs.~(\ref{eq:a_{11}}), (\ref{eq:zero_term}), (\ref{eq:ortho_1}), and (\ref{eq:ortho_2}) that
%its right-hand side becomes zero.
%\[
%(\alpha |b_{10}|^2 + \beta |c_{10}|^2)\left((1-\lambda)|a_{10}|^2-\lambda|a_{01}|^2\right).
%\]
%we obtain that
%\begin{eqnarray}
%(1-2\lambda)|a_{01}|^2(\Tr{S_A}-1) = \nonumber \\
%(\alpha |b_{10}|^2 + \beta |c_{10}|^2)\left((1-\lambda)|a_{10}|^2-\lambda|a_{01}|^2\right)~\label{eq:}.
%\end{eqnarray}
%By Eqs.~(\ref{eq:a_{11}}) and~(\ref{eq:zero_term}), 
%and the asumption that $a_{00} \ne 0$,
%$(1-\lambda)|a_{10}|^2-\lambda|a_{01}|^2 = 0$.
Thus, 
%$(1-2\lambda)|a_{00}|^2|a_{01}|^2(\Tr{S_A}-1) = 0$.
%Since $a_{00}, a_{01} \ne 0$ and $\Tr{S_A}>1$,
$\lambda =1/2$,
and so $\ket{\psi_1}$ is maximally entangled.
Moreover, by Eq.~(\ref{eq:zero_term}) and~(\ref{eq:diag_equality}),
%$|a_{00}|^2+|a_{01}|^2=1/2$, and so
we can see that $\ket{\psi_2}$ is also maximally entangled.
Hence, $\Tr_A{J_{AB}} = I_B$, and
This completes the proof.
\end{proof}

\begin{Rem} %추가한것
It is clear by Lemma~\ref{Lemma:M_0^QNS} that
for nonunital qubit channels,
$\M_0^{\mathrm{QNS}}=1$,
and so 
the one-shot (entanglement-assisted) zero-error classical capacity is also one
for nonunital qubit channels.
%$\M_0 = \M_0^{\mathrm{SE}}=1$.
\end{Rem}

%%%%%%%%%%%%%%%%%%%%%%%%%%%%%%%%%%%%%%%%%%%%
%%                                                       main results                                                    %%
%%%%%%%%%%%%%%%%%%%%%%%%%%%%%%%%%%%%%%%%%%%%
\section{Relations to the entanglement-assisted zero-error classical capacities}
\label{sec:main}
In this section, 
we show that $\M_0^{\mathrm{QNS}}$ is equals to 
the one-shot entanglement-assisted zero-error classical capacity $\M_0^{\mathrm{SE}}$ 
for qubit channels.
This says that 
shared entanglement is sufficient to achieve $\M_0^{\mathrm{QNS}}(\cN)$ 
for qubit channels $\cN$.

It has been known that for a quantum channel $\cN$, 
$\M_0^{\mathrm{SE}}(\cN)$ depends only on its associated subspace 
called the {\em noncommutative graph} of $\cN$~\cite{DSW13}.
%the one-shot (unassisted) zero-error classical capacity $\M_0(\cN)$ and $\M_0^{\mathrm{SE}}(\cN)$ 
%depend only on some associated subspace~\cite{Duan09,DSW13}.
Precisely, the associated subspace $S$ is defined by 
$S\equiv \mathrm{span}\{ E^{\dagger}_i E_j \}$, %\le \cL(\Complex^d)$, 
where $E_i$ are Kraus operators of the channel $\cN$.
%and $\cL(\Complex^d)$ is the space of linear operators on $\Complex^d$.
Thus, to compare $\M_0^{\mathrm{QNS}}$ with $\M_0^{\mathrm{SE}}$, %or $\C_0$, 
we need to reformulate the results in the previous section 
in terms of the noncommutative graphs.

For a quantum channel $\cN$ with Kraus operators $E_i$, 
define the {\em Kraus operator space} $K$ of $\cN$ as
$K \equiv \mathrm{span}\{E_i\}$. %is called the {\em Kraus operator space} of $\cN$.
We note that
$S = \mathrm{span}\{G_j^{\dagger}G_k\}$
for any orthonormal basis $\{G_j\}$ of $K$, and that
there exists an orthonormal basis $\{F_j\}$ for $K$ such that
the CJ matrix $J_{AB}$ of $\cN$ can be expressed as 
$J_{AB} = \sum_{j} a_j(I \otimes F_j){\ket{\Phi}\bra{\Phi}}(I \otimes {F_j}^\dagger)$,
where $a_j > 0$~\cite{DW14}.
% and $\ket\Phi_{AB} = \sum_{k}\ket{k}_A\ket{k}_B$~\cite{DW14}.
%the unnormalized maximally entangled state
Then we obtain the following theorem.

\begin{Thm}\label{Thm:qubit}
For any qubit channel $\cN$, 
let $S$ be the noncommutative graph of $\cN$. 
Then $\M_0^{\mathrm{QNS}}(\cN) = 4/{\dim(S)}$.
\end{Thm}
\begin{proof}
First, let us consider unital qubit channels.
As stated at the beginning of the previous section, 
it is sufficient to consider Pauli channels.
%For a Pauli channel $\cN^{P}$
%with the probability distribution $\{p_{ij}\}$,
%let $k$ be the number of nonzero probabilities $p_{ij}$.
Let $\cN^{P}$ be a Pauli channel
with the probability distribution $\{p_{ij}\}$,
and let $k$ be the number of nonzero probabilities $p_{ij}$.
It is easy to calculate the dimension of the noncommutative graph $S$ of $\cN^{P}$
according to $k$.
Indeed, when $k=1, 2, 3$, and $4$, $\dim(S) = 1, 2, 4$, and $4$, respectively.
By Theorem~\ref{Prop:Pauli ch.}, 
$\M_0^{\mathrm{QNS}}(\cN^{P}) = 4/{\dim(S)}$.

We now consider nonunital qubit channels.
By Theorem~\ref{Prop:Upsilon_nonunital},
it is sufficient to show that $\dim(S)=4$.
Let $J_{AB}$ be the CJ matrix of a nonunital qubit channel $\cN$. 
%and $S$ be the noncommutative graph of $\cN$.
Then $\mathrm{rank}(J_{AB}) > 1$, 
since otherwise the associated channel must be unital
as in the proof of Theorem~\ref{Prop:Upsilon_nonunital}. 
When $\mathrm{rank}(J_{AB})=4$, 
the associated projection $P_{AB}$ is equals to $I$, 
and so it is clear that $\dim(S)=4$.

Assume that $\mathrm{rank}(J_{AB})=2$ or 3.
Let
\[J_{AB} = \sum_{i=1}^3 r_i(I \otimes F_i){\ket{\Phi}\bra{\Phi}}(I \otimes {F_i}^\dagger),
\] 
where 
$r_1, r_2 > 0$, $r_3\ge0$, 
$\sum_{i=1}^3 r_i = 2$, and
$\{F_i\}$ is an orthonormal basis for
the Kraus operator space $K$ of $\cN$. % \equiv \mathrm{span}\{E_i\}$.
Without loss of generality,
we can let
$F_1 = \sqrt{a}\ket{0}\bra{0} + \sqrt{1-a}\ket{1}\bra{1}$, 
$F_2 = \sum_{s,t=0}^1 b_{st}\ket{s}\bra{t}$, and
$F_3 = \sum_{s,t=0}^1 c_{st}\ket{s}\bra{t}$, 
where
$1/2 \le a \le 1$ and
$b_{st}, c_{st} \in \Complex$.
We note that 
when $r_3=0$, 
it becomes the case of that $\mathrm{rank}(J_{AB})=2$ 
, and that
$\Tr_A{J_{AB}} \ne I_B$, 
since $\cN$ is nonunital.

From the orthonormality of $\{F_i\}$,
we obtain the following equalities
\begin{eqnarray}
0 &=& \sqrt{a}b_{00} + \sqrt{1-a}b_{11}, \label{eq:orthog_12} \\
0 &=& \sqrt{a}c_{00} + \sqrt{1-a}c_{11}, \label{eq:orthog_13} \\
0 &=& a(b_{01}^*c_{01}+b_{10}^*c_{10}) + b_{11}^*c_{11}, \label{eq:orthog_23} \\
%\sum_{s,t=0}^1 b_{st}^*c_{st} &=& 0, \label{eq:orthog_23} \\
a &=& a\left(|b_{01}|^2+|b_{10}|^2\right) + |b_{11}|^2,\label{eq:normal_2} \\
a &=& a\left(|c_{01}|^2+|c_{10}|^2\right) + |c_{11}|^2. \label{eq:normal_3}
\end{eqnarray}
Since $\cN$ is a quantum channel, 
$\Tr_B{J_{AB}} = I_A$, 
which gives the following equalities
\begin{eqnarray}
1 = &r_1a + r_2(|b_{00}|^2+|b_{10}|^2) + r_3(|c_{00}|^2+|c_{10}|^2), \label{eq:TP_00} \\
%1-2a^2 &=& r_2(|b_{00}|^2+|b_{10}|^2-a^2) + r_3(|c_{00}|^2+|c_{10}|^2-a^2) \\
0 = &r_2(b_{00}^*b_{01}+b_{10}^*b_{11}) + r_3(c_{00}^*c_{01}+c_{10}^*c_{11}).\label{eq:TP_01}
%0 &=& r_2(b_{00}b_{01}^*+b_{10}b_{11}^*) + r_3(c_{00}c_{01}^*+c_{10}c_{11}^*), \label{eq:TP_10} \\
%1 &=& r_1(1-a^2) + r_2(|b_{01}|^2+|b_{11}|^2) + r_3(|c_{01}|^2+|c_{11}|^2). \label{eq:}
\end{eqnarray}

First, let us consider when $a=1/2$.
By Eqs.~(\ref{eq:orthog_12}) and~(\ref{eq:orthog_13}), 
$b_{00} = -b_{11}$ and $c_{00} = -c_{11}$. 
Then $\Tr_A{J_{AB}} = \Tr_B{J_{AB}} = I$, 
hence this is a contradiction to being nonunital.

%We will show that $S = \cL(\Complex^2)$.
%To do this, 
Let us now assume that $a \ne 1/2$ 
and $\mathrm{rank}(J_{AB})=3$,
that is, $r_3>0$.
We consider the matrix $M_0$ whose columns are $\ket{F_i^{\dagger}F_j}$, 
where
$1 \le i,j \le 3$ and 
$\ket{A} \equiv \sum_{s,t=0}^1 a_{st}\ket{s}\ket{t}$ 
for a matrix $A = \sum_{s,t=0}^1 a_{st}\ket{s}\bra{t}$. 
Then by applying elementary operations properly on $M_0$
we obtain the following matrix

\begin{equation}\label{eq:equiv_matrix}
M=
\left(\begin{array}{c@{}|c}
	\begin{array}{ccc}
	1 & 0 & 0\\
       0 & 1 & 0
	\end{array} & 
	\begin{array}{cccccc}
	0 & 0 & 0 & 0 & 0 & 0 \\
	0 & 0 & 0 & 0 & 0 & 0	
	\end{array} \\ \hline
	\begin{array}{ccc}
	0 & 0 & z\\
       0 & 0 & z^*
	\end{array} & 
	\begin{array}{ccc}
	M_1 & M_2 & M_3	
	\end{array}
\end{array}\right), 
\end{equation}
where 
\begin{eqnarray}
M_1 &=& 
\left(\begin{array}{cc}
\sqrt{a}b_{01} & \sqrt{1-a}b_{10}^* \\
\sqrt{1-a}b_{10} & \sqrt{a}b_{01}^*
\end{array}\right), \\ 
M_2 &=& 
\left(\begin{array}{cc}
\sqrt{a}c_{01} & \sqrt{1-a}c_{10}^* \\
\sqrt{1-a}c_{10} & \sqrt{a}c_{01}^*
\end{array}\right), \\ 
M_3 &=& 
\left(\begin{array}{cc}
b_{00}^*c_{01}+b_{10}^*c_{11} & b_{01}c_{00}^*+b_{11}c_{10}^* \\
b_{01}^*c_{00}+b_{11}^*c_{10} & b_{00}c_{01}^*+b_{10}c_{11}^*
\end{array}\right), \\%and
z &=& c_{00}^*c_{01}+c_{10}^*c_{11}.
\end{eqnarray}

We will show that $\mathrm{rank}(M)=4$, 
that is, $\dim(S) = 4$.
To do this, 
we suppose that $\mathrm{rank}(M)<4$. 
Then we note that any two-by-two submatrix of the lower-right block of $M$
consisting of $M_1$, $M_2$, and $M_3$
has zero determinant.

Assume that $a=1$.
By Eqs.~(\ref{eq:orthog_12}) and~(\ref{eq:orthog_13}), 
$b_{00}=0=c_{00}$.
Since $\det(M_3)=0$, 
\begin{equation}\label{eq:from_M3}
|b_{10}|^2|c_{11}|^2=|b_{11}|^2|c_{10}|^2.
\end{equation}
Since $\det(M_1)=0$ and $\det(M_2)=0$, 
%$b_{01}=0$ and $c_{01}=0$.
from Eqs.~(\ref{eq:orthog_23}),~(\ref{eq:normal_2}), and~(\ref{eq:normal_3}), 
\begin{eqnarray}
|b_{10}|^2|c_{10}|^2 &=& |b_{11}|^2|c_{11}|^2, \label{eq:from_orthog} \\
|b_{10}|^2+|b_{11}|^2 &=& |c_{10}|^2+|c_{11}|^2 = 1. \label{eq:from_normal}
\end{eqnarray}
%from Eq.~(\ref{eq:orthog_23}), 
%\begin{equation}\label{eq:from_orthog}
%|b_{10}|^2|c_{10}|^2=|b_{11}|^2|c_{11}|^2,
%\end{equation}
%and from Eqs.~(\ref{eq:normal_2}) and~(\ref{eq:normal_3}), 
%\begin{equation}\label{eq:from_normal}
%|b_{10}|^2+|b_{11}|^2=1=|c_{10}|^2+|c_{11}|^2.
%\end{equation}
%
It is straightforward 
from Eqs.~(\ref{eq:from_M3}),~(\ref{eq:from_orthog}), and~(\ref{eq:from_normal}) 
to obtain the following equalities
\begin{equation}\label{eq:b10=1/2=c10}
|b_{10}|^2 = 1/2 = |c_{10}|^2.
\end{equation}
Then, from Eq.~(\ref{eq:TP_00}), 
$r_1=0$.
This is a contradiction,
and hence $\mathrm{rank}(M)=4$.

Assume that $1/2<a<1$.
%Suppose that $\mathrm{rank}(M)<4$. 
%
Since $\det(M_1)=0$ and $\det(M_2)=0$, 
\begin{eqnarray}
|b_{01}|^2 &=& \frac{1-a}{a}|b_{10}|^2, \label{eq:b01=} \\
|c_{01}|^2 &=& \frac{1-a}{a}|c_{10}|^2. \label{eq:c01=}
\end{eqnarray}
Then, from~(\ref{eq:normal_2}) and~(\ref{eq:normal_3}), 
\begin{equation}\label{eq:10^2+11^2=a^2}
|b_{10}|^2 + |b_{11}|^2 = a = |c_{10}|^2 + |c_{11}|^2.
\end{equation}
Since 
$\det\left(\begin{array}{cc}
\sqrt{a}b_{01} & \sqrt{a}c_{01} \\
\sqrt{1-a}b_{10} & \sqrt{1-a}c_{10}
\end{array}\right)=0$, 
\begin{equation}\label{eq:det_12_13}
b_{01}c_{10}-b_{10}c_{01} = 0,
\end{equation}
and since 
$\det\left(\begin{array}{cc}
\sqrt{a}b_{01} & \sqrt{1-a}c_{10}^* \\
\sqrt{1-a}b_{10} & \sqrt{a}c_{01}^*
\end{array}\right)=0$, 
from Eq.~(\ref{eq:orthog_23}), 
\begin{equation}\label{eq:det_12_31}
%a^2b_{01}c_{01}^* = (1-a^2)b_{10}c_{10}^*
b_{10}c_{10}^* + b_{11}c_{11}^* = 0.
\end{equation}
Since 
$\det\left(\begin{array}{cc}
\sqrt{a}b_{01} & b_{00}^*c_{01}+b_{10}^*c_{11} \\
\sqrt{1-a}b_{10} & b_{01}^*c_{00}+b_{11}^*c_{10}
\end{array}\right)=0$, 
by Eqs.~(\ref{eq:orthog_12}),~(\ref{eq:orthog_13}),~(\ref{eq:b01=}), and~(\ref{eq:det_12_13}), 
\begin{equation}\label{eq:det_12_23}
\sqrt{1-a}|b_{10}|^2c_{11} = \sqrt{a}b_{10}b_{11}^*c_{01}.
\end{equation}
%
%Taking the absolute value on both sides of Eq.~(\ref{eq:det_12_23}), 
Then, from Eqs.~(\ref{eq:c01=}) and~(\ref{eq:10^2+11^2=a^2}), 
we obtain
\begin{equation}\label{eq:zero_cond}
|b_{10}|(|b_{10}|-|c_{10}|) = 0.
\end{equation}
Thus, $b_{10}=0$ or $|b_{10}|=|c_{10}|$.

Suppose that $b_{10}=0$. 
By Eqs.~(\ref{eq:orthog_12}),~(\ref{eq:orthog_13}),~(\ref{eq:10^2+11^2=a^2}), and~(\ref{eq:det_12_31}), 
it is not hard to show that 
$|b_{00}|^2=1-a$ and $|c_{00}|^2+|c_{10}|^2=a$.
From Eq.~(\ref{eq:TP_00}), 
$r_3=0$. 
This is a contradiction,
and so $b_{10}\ne0$ and $|b_{10}|=|c_{10}|$.
%
%By Eqs.~(\ref{eq:b01=}),~(\ref{eq:c01=}), and~(\ref{eq:10^2+11^2=a^2}),  
%\begin{equation}\label{eq:coeff_equal}
%|b_{st}| = |c_{st}|
%\end{equation}
%for any $s,t \in \{0,1\}$.
By Eqs.~(\ref{eq:b01=}),~(\ref{eq:c01=}),~(\ref{eq:10^2+11^2=a^2}), and~(\ref{eq:det_12_31}), 
$|b_{st}|^2 = a/2 = |c_{st}|^2$
for any $s,t \in \{0,1\}$.
From Eq.~(\ref{eq:TP_00}), $a=1/2$.
This is a contradiction, 
and so $\textrm{rank}(M)=4$.

We now assume that $a \ne 1/2$ and $\mathrm{rank}(J_{AB})=2$ 
in the rest of this proof.
As in the case of that $\mathrm{rank}(J_{AB})=2$, 
let us consider the matrix whose columns are $\ket{F_i^{\dagger}F_j}$, 
where
$1 \le i,j \le 2$. 
%and 
%$\ket{A} \equiv \sum_{s,t=0}^1 a_{st}\ket{s}\ket{t}$ 
%for a matrix $A = \sum_{s,t=0}^1 a_{st}\ket{s}\bra{t}$. 
Then by applying elementary operations properly
we obtain the following matrix

\begin{equation}\label{eq:equiv_matrix}
N=
\left(\begin{array}{cc|cc}
1 & 0 & 0 & 0 \\
0 & 1 & 0 & 0 \\ \hline
0 & 0 & \sqrt{a}b_{01} & \sqrt{1-a}b_{10}^* \\
0 & 0 & \sqrt{1-a}b_{10} & \sqrt{a}b_{01}^* 
\end{array}\right),
%\left(\begin{array}{c|c}
%	\begin{array}{cc}
%	1 & 0 \\
%       0 & 1 
%	\end{array} & 
%	\begin{array}{cc}
%	0 & 0 \\
%	0 & 0	
%	\end{array} \\ \hline
%	\begin{array}{cc}
%	0 & 0 \\
%       0 & 0
%	\end{array} & 
%	\begin{array}{c@{}c} %c@{}c
%	ab_{01} & \sqrt{1-a^2}b_{10}^* \\
%	\sqrt{1-a^2}b_{10} & ab_{01}^*
%	\end{array}
%\end{array}\right), 
\end{equation}

Suppose that $\mathrm{rank}(N)<4$.
Then since $\det(N)=0$, 
\begin{equation}\label{eq:det=0}
a|b_{01}|^2 - (1-a)|b_{10}|^2 = 0.
\end{equation}
By Eq.~(\ref{eq:TP_01}), 
\begin{equation}\label{eq:from_TP}
|b_{00}|^2|b_{01}|^2=|b_{10}|^2|b_{11}|^2.
\end{equation}
Then from Eqs.~(\ref{eq:orthog_12}),~(\ref{eq:normal_2}), and~(\ref{eq:det=0}), 
\begin{equation}
|b_{10}|^2(a-|b_{10}|^2) = 0.
\end{equation}
If $b_{10}=0$, 
by Eq.~(\ref{eq:det=0}), 
$b_{01}=0$, 
and so $\Tr_A{J_{AB}} = \Tr_B{J_{AB}} = I$.
This is a contradiction, 
and hence $|b_{10}|^2=a$.
By Eqs.~(\ref{eq:orthog_12}),~(\ref{eq:normal_2}), and~(\ref{eq:det=0}), 
$|b_{00}|^2+|b_{10}|^2=a$.
Then by Eq.~(\ref{eq:TP_00}), 
$a=1/2$, 
which implies being unital.
Therefore, we can conclude that $\mathrm{rank}(N)=4$, 
and hence $\dim(S)=4$.
\end{proof}

%% C_0^SE for qubit channels %%
Since $\M_0^{\mathrm{SE}}(\cN)$ depends on 
the noncommutative graph $S$ of the channel $\cN$, 
$\M_0^{\mathrm{SE}}(\cN)$ can be denoted by $\M_0^{\mathrm{SE}}(S)$.
For quantum channels with qubit inputs, 
that is, noncommutative graphs $S \subset \cL(\Complex^2)$, 
$\M_0^{\mathrm{SE}}(S)$ and 
the (asymptotic) entanglement-assisted zero-error classical capacity $\C_0^{\mathrm{SE}}(S)$ 
can be obtained from Ref.~\cite{DSW13} 
as in the following proposition.
\begin{Prop}\label{Prop:M0^SE_qubit}
For a noncommutative graph $S \subset \cL(\Complex^2)$, 
$\C_0^{\mathrm{SE}}(S) = \log\M_0^{\mathrm{SE}}(S)$.
Moreover, when $\dim(S)=1,2,3$, and $4$, 
$\M_0^{\mathrm{SE}}(S)=4,2,2$, and $1$, respectively.
\end{Prop}

%% M_0^SE = M_0^QNS for qubit channels %%
In fact, as shown in the proof of Theorem~\ref{Thm:qubit}, 
$\dim(S) \ne 3$ for qubit channels.
From Theorem~\ref{Thm:qubit} and Proposition~\ref{Prop:M0^SE_qubit}, 
we directly obtain the following corollary.
\begin{Cor}\label{Cor:result}
For any qubit channel $\cN$, 
$\M_0^{\mathrm{SE}}(S) = \M_0^{\mathrm{QNS}}(\cN)$ and
$\C_0^{\mathrm{SE}}(\cN) = \log\M_0^{\mathrm{QNS}}(\cN)$.
\end{Cor}

\begin{Rem}
Corollary~\ref{Cor:result} means that 
QNS correlations are not more powerful than shared entanglement 
for a single use of a qubit channel.
\end{Rem}

%%%%%%%%%%%%%%%%%%%%%%%%%%%%%%%%%%%%%%%%%%%%
%%                                               Conclusion                                                              %%
%%%%%%%%%%%%%%%%%%%%%%%%%%%%%%%%%%%%%%%%%%%%
\section{Conclusion}
\label{sec:conclusion}
%% summary %%
We have considered 
the one-shot QNS assisted zero-error classical capacities $\M_0^{\mathrm{QNS}}$ 
of qubit channels.
First, we have calculated the exact values of $\M_0^{\mathrm{QNS}}$
for Pauli channels and nonunital qubit channels, 
and then %rewrite $\M_0^{\mathrm{QNS}}$ in terms of the noncommutative graphs.
we have shown that 
$\M_0^{\mathrm{SE}}(S) = \M_0^{\mathrm{QNS}}(\cN)$ 
%and $\C_0^{\mathrm{SE}}(\cN) = \log\M_0^{\mathrm{QNS}}(\cN)$ 
for any qubit channel.

%% examples %%
Moreover, 
we can present examples of quantum channels $\cN$ such that
$\C_0^{\mathrm{SE}}(\cN) = 0 < \C_0^{\mathrm{QNS}}(\cN)$.
(Such examples for classical channels were already known~\cite{CLMW11}).
This implies that
QNS correlations can make useless channels useful
although the shared entanglement cannot do so.
Let us consider Pauli channels with three nonzero probability $p_j$.
%Indeed, these are unital qubit channels.
Then Theorem~\ref{Prop:Pauli ch.} and Corollary~\ref{Cor:QNS_Pauli} shows that
$\log\M_0^{\mathrm{QNS}} = 0 < \log(4/3) = \C_0^{\mathrm{QNS}}$.
By Corollary~\ref{Cor:result}, 
we can obtain the inequality
$\C_0^{\mathrm{SE}} = 0 < \C_0^{\mathrm{QNS}}$.
%for Pauli channels with three nonzero probability $p_j$.

Another example is the (nonunital) extremal qubit channel 
$\cN(\cdot) = \sum_{j=1}^{2}E_j \cdot E_j^{\dagger}$,
where $E_1 = \cos\theta\ket{0}\bra{0} + \cos\varphi\ket{1}\bra{1}$, 
$E_2 = \sin\varphi\ket{0}\bra{1} + \sin\theta\ket{1}\bra{0}$, 
and $\theta, \varphi \in \Real$ with $\cos^2\theta \ne \cos^2\varphi$~\cite{{RSW02}}.
By the criterion for determining whether 
$\C_0^{\mathrm{QNS}}=0$ in Ref.~\cite{{DW14}}, 
we can see that
$\C_0^{\mathrm{QNS}}>0$, while $\C_0^{\mathrm{SE}} = 0$ 
by Theorem~\ref{Prop:Upsilon_nonunital} and Corollary~\ref{Cor:result}.
 
%% future woks %%
There are many questions about $\C_0^{\mathrm{QNS}}$.
One of them is determining the value of $\C_0^{\mathrm{QNS}}$ for quantum channels.
In fact, 
the explicit value of $\C_0^{\mathrm{QNS}}$ is yet unknown, 
even for the above nonunital extremal qubit channels.
Another is the additivity of $\C_0^{\mathrm{QNS}}$.
For any unital qubit channel, 
$\C_0^{\mathrm{QNS}}$ is additive by Corollary~\ref{Cor:QNS_Pauli}.
Moreover, 
$\C_0^{\mathrm{QNS}}$ is additive for any classical-quantum channel~\cite{DW14}.
However, 
it is not known that the additivity holds for any quantum channel.
Although the above questions are proved for unital qubit channels in this paper, 
our work could shed light on the mentioned problem.

\acknowledgments
We thank Runyao Duan for helpful comments.
This research was supported by the Basic Science Research Program through the National Research Foundation of Korea funded by the Ministry of Education (NRF-2012R1A1A2003441).

\bibliography{QNS}

\end{document}